\documentclass[11pt]{article}
\usepackage{pdfsync}
\usepackage{amsmath}
\usepackage{amssymb}
\usepackage{amscd}
\usepackage{mathrsfs}
\usepackage{color}
\usepackage{ulem}
\usepackage{bm}
\usepackage{graphicx}
\usepackage{mathtools}
\usepackage{hyperref}

\newtheorem{theorem}{Theorem}[section]
\newtheorem{prop}[theorem]{Proposition}
\newtheorem{cor}[theorem]{Corollary}
\newtheorem{lemma}[theorem]{Lemma}
\newtheorem{algorithm}[theorem]{Algorithm}
\newtheorem{remark}[theorem]{Remark}

\newtheorem{define}[theorem]{Definition}
\newtheorem{example}[theorem]{Example}
\newtheorem{problem}[theorem]{Problem}

\newcommand{\ord}{\mbox{\rm ord}}

\definecolor{highlight}{rgb}{.5,0,.5}

\newcommand{\zero}{{\mbox{\rm Zero}}}

\newcommand{\gal}{{\mbox{\rm Gal}}}

\def\L{\hbox{\bf L}}

\def\ring{{k[Y, 1/\det(Y)]}}
\def\GLbQ{{\GL_n(\overline{\bQ})}}

\def\diag{\hbox{\rm diag}}
\def\GL{{\rm GL}}
\def\stab{{\rm stab}}
\def\Zero{{\rm Zero}}

\def\P{{\mathcal P}}

\def\L{{\mathcal L}}

\def\H{{\mathcal H}}
\def\X{{\mathcal X}}

\def\bZ {{\mathbb{Z}}}
\def\bS{{\mathbb{S}}}

\def\bP{{\mathbf{P}}}
\def\bQ{{\mathbb{Q}}}
\def\barQ{{\overline{\mathbb{Q}}}}

\def\calZ{{\mathcal Z}}

\def\calS{{\cal S}}
\def\bfalpha{{\bm\alpha}}
\def\bfa{{\mathbf a}}
\def\bfb{{\mathbf b}}
\def\bfc{{\mathbf c}}

\def\bfv{{\mathbf v}}
\def\bfw{{\mathbf w}}

\def\bfm{{\mathbf m}}

\def\bfh{{\mathbf h}}

\def\bfZ{{\mathbf Z}}

\makeatletter
\newcommand{\rmnum}[1]{\rm\romannumeral #1}
\newcommand{\Rmnum}[1]{\expandafter\@slowromancap\romannumeral #1@}
\makeatother

\begin{document}

\title{On the Computation of\\ the Galois Group of Linear Difference Equations}
\author{Ruyong Feng\footnote{ryfeng@amss.ac.cn. This work is partially supported by a National Key Basic Research Project of China (2011CB302400) and
by a grant from NSFC (60821002).} \\ KLMM, AMSS, Chinese Academy of Sciences, \\Beijing 100190, China
}
\date{} \maketitle

\begin{abstract} We present an algorithm that determines the Galois group of linear difference equations with rational function coefficients.   \end{abstract}


\section{Introduction}
\label{sec-preliminary}
The current algorithms for computing the Galois group of linear difference equations were only valid for the equations of special types, such as the second order equations, the equations of diagonal form or with constant coefficients and so on. In \cite{hendriks}, a difference analogue of Kovacic's algorithm was developed for linear difference equations of order two.
In \cite{put-singer}, algorithms for linear difference equations of diagonal form were developed. For linear difference equations with constant coefficients, an algorithm can be found in \cite{singer2}, where the author further showed that there is a recursive procedure that derives the Galois group from the ideal of algebraic relations among solutions, and vice versa. In \cite{maier}, Maier gave upper and lower bounds for the Galois groups of Frobenius difference equations over $(\mathbb{F}_q(s,t),\phi_q)$, where $\phi_q(s)=s^q$ and $\phi_q(a)=a$ for all $a\in \mathbb{F}_q(t)$. On the contrary, algorithms for computing the Galois groups of linear differential equations have been well-developed (see \cite{compoint-singer, kovacic,singer-ulmer,hrushovski}). Particularly, in \cite{hrushovski}, Hrushovski developed an algorithm that calculates the Galois groups of all linear differential equations with rational function coefficients. His algorithm involved many arguments from logical language and has recently been reworked by Rettstadt in \cite{rettstadt} and by the author in \cite{feng}. Here, in this paper, we develop an algorithm for computing the Galois group of linear difference equations with rational function coefficients of arbitrary order. Our algorithm can be considered as a difference analogue of Hrushovski's algorithm.

The philosophy of computing the Galois groups of linear difference equations is quite similar to that of linear differential equations.
The Galois groups of these two kinds of equations are linear algebraic groups over the field of constants. Hence bounds for the defining equations of linear algebraic groups developed for the differential case can be applied to the difference case without any modification. However, there exist some results in differential algebra whose difference analogues are not correct any more, and vice versa. For example, associated primes of a radical differential ideal are again differential ideals, while those of a radical $\sigma$-ideal need not be $\sigma$-ideals but $\sigma^\delta$-ideals for some integer $\delta$. This forces us to consider $\sigma^\delta$-ideals. Another example is that the Picard-Vessiot extension ring for linear differential equations is not necessarily the coordinate ring of a trivial torsor for the Galois group, while that for linear difference equations is the coordinate ring of a trivial torsor. This implies that one only needs to consider objects such as hypergeometric elements that are defined over the basic field.

Throughout this paper, $k$ stands for the field of rational functions in $x$ with coefficients in $\overline{\bQ}$, the algebraic closure of the field of rational numbers, and $\bar{k}$ stands for its algebraic closure. The difference field which we are interested in is the field $k$ with an automorphism $\sigma$ given by $\sigma(x)=x+1$ and $\sigma(c)=c$ for $c\in \overline{\bQ}$. Consider the following linear difference equations
\begin{equation}
\label{EQ:differenceeqns}
 \sigma(Y)=AY
\end{equation}
where $Y$ is an $n\times 1$ vector with indeterminate entries and $A\in \GL_n(k)$. Let $R$ be the Picard-Vessiot extension ring of $k$ for (\ref{EQ:differenceeqns}). The Galois group of (\ref{EQ:differenceeqns}) over $k$, denoted by $\gal(R/k)$, is defined to be the set of $\sigma$-$k$-automorphisms of $R$, i.e. $k$-automorphisms of $R$ that commute with $\sigma$. Let $F$ be a fundamental matrix of (\ref{EQ:differenceeqns}) with entries in $R$, i.e. $F\in \GL_n(R)$ satisfying $\sigma(F)=AF$. Then for any $\phi\in\gal(R/k)$, $\phi(F)$ is another fundamental matrix of (\ref{EQ:differenceeqns}). Thus there exists $[\phi]\in \GL_n(\barQ)$ such that $\phi(F)=F[\phi]$. The map given by $\phi\rightarrow [\phi]$ is a group homomorphism of $\gal(R/k)$ into $\GL_n(\barQ)$. Denote by $G$ the set $\{[\phi]\,\,|\,\,\phi\in\gal(R/k)\}$. It was proved in (Theorem 1.13, page 11 of \cite{put-singer}) that $G$ is a linear algebraic group defined over $\barQ$. The reader is referred to Chapter 1 of \cite{put-singer} for more information about the Galois theory of linear difference equations.

The group $G$ can be reformulated as the stabilizer of some ideal in a $\sigma$-ring, which we describe below. Let $Y$ denote an $n\times n$ matrix $(y_{i,j})$, where the $y_{i,j}$ are indeterminates. Sometimes, in brief, we also consider $Y$ as a set of indeterminates. By setting $\sigma(Y)=AY$, one can extend $\sigma$ from $k$ to $\ring$ so that it becomes a difference extension ring of $k$. The results in Section 1.1 of \cite{put-singer} imply that $R$ is isomorphic to $\ring/I$ for some maximal $\sigma$-ideal $I$. Define an action of $\GLbQ$ on $\ring$ given by $g\cdot Y=Yg$ for all $g\in \GLbQ$. Suppose that $J$ is an ideal of $\ring$.
The {\it stabilizer} of $J$, denoted by $\stab(J)$, is defined as
$$
   \stab(J)=\{g\in \GLbQ\,\,|\,\, P(Yg)\in J,\,\,\forall\,\,P\in J\},
$$
which is an algebraic subgroup of $\GLbQ$. Set
 $$
   I_F=\left\{P\in \ring\,\,|\,\,P(F)=0\right\}.
 $$
 Then $I_F$ is a maximal $\sigma$-ideal and $G=\stab(I_F)$. By the uniqueness of the Picard-Vessiot extension ring of $k$ for (\ref{EQ:differenceeqns}), one sees that for any maximal $\sigma$-ideal $I$ of $\ring$, there is $g\in \GL_n(\barQ)$ such that $g\cdot I=I_F$.
 From this, one can readily verify the stabilizers of maximal $\sigma$-ideals in $\ring$ are conjugated. In other words, as linear algebraic groups, these stabilizers are isomorphic. Therefore we shall also call the stabilizer of a maximal $\sigma$-ideal of $\ring$ the Galois group of (\ref{EQ:differenceeqns}) over $k$.
Using the Gr\"{o}bner base method, one can obtain the defining equations of $\stab(I)$ easily once a Gr\"{o}bner basis of $I$ is known.
Therefore, the above definition indicates that finding a maximal $\sigma$-ideal of $\ring$ will suffice to determine the Galois group. We shall give in this paper an algorithm that computes a maximal $\sigma$-ideal of $\ring$.

The rest of the paper is organized as follows. In Section 2, we introduce some basic results that provide the theoretical background of our algorithm. Meanwhile, we introduce some basic definitions such as proto-groups, proto-maximal $\sigma$-ideals and so on. In Section 3, we show how to compute a proto-maximal $\sigma$-ideal. In Section 4, we describe a method to extend a proto-maximal $\sigma$-ideal to a maximal $\sigma^\delta$-ideal so that one can easily obtain a maximal $\sigma$-ideal by taking the intersection of ideals. In Section 5, the methods developed in the previous sections are summarized as an algorithm, and an example is presented to illustrate the algorithm. In Appendix A, we describe a method to find coefficient bounds for generators of a proto-maximal $\sigma$-ideal. In Appendix B, an algorithm for computing $\sigma^\delta$-hypergeometric elements in $\ring/I_{\rm irr}$ is developed, where $I_{\rm irr}$ is a prime $\sigma^\delta$-ideal.
\section{Some basic results}
In this section, we shall introduce some basic results about proto-groups, $k$-torsors and several related problems whose algorithmic solutions will be needed in our algorithm.
\subsection{Proto-groups}
As in the differential case, bounds on algebraic subgroups of $\GL_n(\barQ)$ play a central role in the main algorithm presented in this paper. Let $H$ be an algebraic subgroup of $\GLbQ$. For the ease of notation, we shall use $H(k)$ (resp. $H(\bar{k})$) to denote $k$-points (resp. $\bar{k}$-points) of $H$.
We shall say $H$ is bounded by a positive integer $d$ if there is a set $\bS\subseteq \barQ[Y]$ such that $H$ is the set of zeroes of $\bS$ in $\GLbQ$ and elements of $\bS$ are of degree not greater than $d$. In brief, $H_u$ stands for the algebraic subgroup of $H$ generated by unipotent elements and $H^\circ$ denotes the identity component of $H$.
\begin{define}
   Let $G, H$ be two algebraic subgroups of $\GLbQ$. $H$ is said to be a proto-group of $G$ if they satisfy the following condition
   $$
       H_u \leq G^\circ \leq G \leq H .
   $$
   In the case that $G$ is the Galois group of (\ref{EQ:differenceeqns}) over $k$, $H$ is called a proto-Galois group of (\ref{EQ:differenceeqns}).
\end{define}
\begin{remark}
   \label{RM:protogroup}
       Suppose that $H$ is a proto-group of $G$ and $\bar{H}$ is an algebraic subgroup satisfying $G\leq \bar{H} \leq H$.
      Since $\bar{H}_u \leq H_u$, one sees that $\bar{H}$ is also a proto-group of $G$.
   \end{remark}
For the convenience, we introduce the following definition.
\begin{define}
  A $\sigma$-ideal $I$ in $\ring$ is called {\it proto-maximal} if $\stab(I)$ is a proto-Galois group of (\ref{EQ:differenceeqns}).
\end{define}
The key point of Hrushovski's algorithm is the following proposition, which is also the core of our algorithm.
\begin{prop}
\label{PROP:bounds}(Corollary 3.7 of \cite{hrushovski}, Corollary B.15 of \cite{feng})
 One can find an integer $\tilde{d}$ only depending on $n$ such that for any algebraic subgroup $G$ of $\GLbQ$, there is a proto-group of $G$ bounded by $\tilde{d}$. Particularly, given linear differential equations, there exists a proto-Galois group of it bounded by the integer $\tilde{d}$.
\end{prop}
The integer $\tilde{d}$ can be explicitly given as follows (see Corollary B.15 of \cite{feng} for details).
Set
\begin{equation}
\label{EQ:bounds1}
   \kappa_1=\max_i\left\{{{n^2+(2n)^{3\cdot8^{n^2}} \choose n^2} \choose i}^2\right\},\quad \quad
        \kappa_2=\kappa_1(2n)^{3\cdot8^{n^2}}{n^2+(2n)^{3\cdot8^{n^2}} \choose n^2}
\end{equation}
and
$$
     \kappa_3=\kappa_2(\kappa_1^2+1)\max_i\left\{{\kappa_1^2+1 \choose i}\right\},\quad I(n)=J\left(\max_i\left\{{n^2+1 \choose i}^2\right\}\right)
$$
where $J(m)$ is a Jordan bound, which is not greater than $\left(\sqrt{8m}+1\right)^{2m^2}-\left(\sqrt{8m}-1\right)^{2m^2}$. Then
\begin{equation}
\label{EQ:bounds2}
   \tilde{d}=(\kappa_3)^{I(n)-1}.
\end{equation}

It is well-known in the theory of linear algebraic groups that any algebraic subgroup of a diagonalizable group $D$ is the intersection of kernels of some characters of $D$ (see Proposition in the page 103 of \cite{humphreys}). Given a connected algebraic group $H$, the following proposition describes algebraic subgroups that are the intersections of some characters of $H$.
\begin{prop}
\label{PROP:kernels}
  Suppose that $H$ is a connected algebraic subgroup of $\GLbQ$. Then $G$ is the intersection of kernels of some characters of $H$ if and only if $H$ is a proto-group of $G$.
\end{prop}
\begin{proof}
  Assume that $H$ is a proto-group of $G$. Let $\chi_1,\cdots,\chi_\ell$ be generators of $X(H)$, the group of characters of $H$. Define a map $\psi: H \rightarrow (\barQ^{\times})^\ell$ given by $\psi(h)=(\chi_1(h),\cdots, \chi_\ell(h))$, where $\barQ^{\times}$ denotes the multiplicative group of $\barQ$. Then $\psi(H)$ is a diagonalizable group and $\psi(G)$ is one of its algebraic subgroups. Due to Proposition in the page 103 of \cite{humphreys}, $\psi(G)$ is the intersection of kernels of some characters of $\psi(H)$. Denote these characters by $\bar{\chi}_1,\cdots,\bar{\chi}_l$. Notice that $\psi$ induces a group homomorphism
  \begin{align*}
      \psi^*: X\left((\barQ^{\times})^\ell\right) &\rightarrow X(H)\\
      \chi&\rightarrow \chi\circ\psi.
  \end{align*}
  We claim that
  $G=\cap_{i=1}^l \ker\left(\psi^*(\bar{\chi}_i)\right)$. Obviously, $G\subseteq \cap_{i=1}^l \ker\left(\psi^*(\bar{\chi}_i)\right)$. Suppose that $h\in \cap_{i=1}^l \ker\left(\psi^*(\bar{\chi}_i)\right)$. Then $\bar{\chi_i}(\psi(h))=1$ for all $1\leq i\leq l$. This implies that $\psi(h)\in \psi(G)$. Lemma B.9 of \cite{humphreys} states that $H_u=\ker(\psi)$. Hence $\ker(\psi)\subseteq G$ and then $h\in G$.

  Conversely, $G$ is the intersection of some characters of $H$. Then $H_u=\ker(\psi)\subseteq G$. Since $H_u$ is connected, $H_u\subseteq G^\circ$. Thus $H$ is a proto-group of $G$.
\end{proof}
The connection between proto-groups and $\sigma$-ideals in $\ring$ is the geometric objects so called $k$-torsors, which are introduced in the next section.

\subsection{$k$-Torsors}
We shall use $\Zero(J)$ to denote the set of zeroes of $J$ in $\GL_n(\bar{k})$, where $J$ is a subset of $\ring$. Suppose that $Z\subseteq \GL_n(\bar{k})$ is a variety defined over $k$. We shall use $I_k(Z)$ to denote the set of all polynomials in $\ring$ that vanish on $Z$.
\begin{define}(see Definition 3.13 of \cite{singer2})
 Let $Z\subseteq \GL_n(\bar{k})$ be a variety defined over $k$ and $H$ an algebraic subgroup of $\GL_n(\bar{k})$ defined over $k$. $Z$ is said to be a $k$-torsor for $H$ if for any $z_1,z_2\in Z$, there is a unique $h\in H$ such that $z_1=z_2h$. A $k$-torsor $Z$ for $H$ is said to be trivial if $Z\cap \GL_n(k)\neq \emptyset$, i.e. $Z=BH$ for some $B\in \GL_n(k)$.
\end{define}
Let $I$ be a maximal $\sigma$-ideal of $\ring$. Then one has that
\begin{prop}
\label{PROP:torsors}(Proposition 1.20, page 15 of \cite{put-singer})
  $\Zero(I)$ is a trivial $k$-torsor for $\stab(I)$.
\end{prop}
Suppose that $H$ is a connected algebraic subgroup of $\GL_n(\bar{k})$, which is defined over $\barQ$, and $Z$ is a trivial $k$-torsor for $H$.
Then for any $B\in Z\cap \GL_n(k)$, the map given by
\begin{align}
 \ring/I_k(H) &\rightarrow \ring/I_k(Z)\label{EQ:map}\\
 P(Y)& \rightarrow P(B^{-1}Y)\notag
\end{align}
is an isomorphism of $k$-algebras.
A theorem of Rosenlicht (\cite{magid,rosenlicht,singer}) implies that invertible regular functions on $Z$  are closely related to characters of $H$. This theorem states: let $H$ be a connected linear algebraic group defined over an algebraically closed field $\bar{k}$ and $y$ be a regular function on $H$ with $1/y$ also a regular function, then $y$ is a $\bar{k}$ multiple of a character.
Notice that characters of $H$ can be viewed as elements in $\barQ[Y,1/\det(Y)]/I_\barQ(H)$.
\begin{lemma}
\label{LM:invertibleelements} Suppose that $J$ is a prime $\sigma$-ideal of $\ring$ and $\Zero(J)$ is a trivial $k$-torsor for $H$. Let $B\in \Zero(J)\cap \GL_n(k)$. If $\chi$ is a character of $H$, then $\chi(B^{-1}Y)$ is invertible in $\ring/J$. Conversely, if $P$ is an invertible element in $\ring/J$, then $P=r \chi(B^{-1}Y)$ for some $r\in k$ and some character $\chi$ of $H$.
\end{lemma}
\begin{proof}
We only need to prove the second assertion. Since $\barQ$ is algebraically closed, $H$ viewed as a linear algebraic group defined over $\bar{k}$ is still connected. The map (\ref{EQ:map}) implies that $P(BY)$ is invertible in $\ring/I_k(H)$. Applying the above theorem of Rosenlicht to $P(BY)$, one has that $P(BY)=r\chi$ for some $r\in \bar{k}$ and some character $\chi$. Observe that $\ring/J$ is a $\sigma$-extension ring of $k$. Due to Lemma 1.19 in the page 15 of \cite{put-singer}, $(\ring/J)\cap \bar{k}=k$. Hence $(\ring/I_k(H))\cap\bar{k}=k$. We then conclude that $r\in k$ and $P=r\chi(B^{-1}Y)$.
\end{proof}
In Section 4, one will see that invertible elements of $\ring/J$ are actually $\sigma$-hypergeometric over $k$. In the case that $J$ is a proto-maximal $\sigma$-ideal, algebraic relations among these $\sigma$-hypergeometric elements will reveal the characters of $H$ that determine the Galois group $G$.
\subsection{Some related problems}
In this paper, we shall need the algorithmic solutions of the following problems.
\begin{itemize}
  \item [$(P1)$]Given an ideal in $k[Y]$, compute a Gr\"{o}bner basis of it with respect to some
       monomial ordering. The reader is referred to Section 2.7 of \cite{cox-little-oshea} and Section 5.5 of \cite{becker-Weispfenning} for the algorithms.
  \item [$(P2)$] Given an unmixed ideal in k[Y], compute its radical and its associated primes. There are several methods for this problem, for instance the methods presented in \cite{gianni-trager-zacharias}, Section 4 of \cite{eisenbud-huneke-vasconcelos}, Section 8.7 of \cite{becker-Weispfenning}, parts 36 and 42 of \cite{seidenberg}.
  \item [$(P3)$] Compute the Galois group of linear difference equations of diagonal form. Equivalently, given $b_1,\cdots, b_\ell \in k$, compute a set of generators of the following $\bZ$-module:
      $$
         \left\{ (z_1,\cdots, z_\ell)\in \bZ^{\ell}\,\,\left|\,\, \exists\,\,f\in k\,\,\mbox{s.t.}\,\,\prod_{i=1}^\ell b_i^{z_i}=\frac{\sigma(f)}{f}\right.\right\}.
      $$
      When $k=\bQ(x)$, a method was described in Section 2.2 of \cite{put-singer}. Using the results in Section 3.2 of \cite{derksen-jeandel-koiran}, one can adapt the method in \cite{put-singer} to solve the problem with $k=\barQ(x)$.
      This problem is the bottleneck in extending our algorithm to the equations over a larger basic field.
  \item [$(P4)$] Give linear difference equations with coefficients in $k$, compute all hypergeometric solutions. The reader is referred to (\cite{cluzeau-van-hoeij, petkosevek}) for algorithms.
\end{itemize}
\section{The computation of a proto-maximal $\sigma$-ideal}
Let $F$ be a fundamental matrix of (\ref{EQ:differenceeqns}) and let $d$ be a positive integer or $\infty$. Denote
\begin{equation}
\label{EQ:definable}
  I_{F,d}=\langle\{ P(Y)\in k[Y]_{\leq d}\,\,|\,\, P(F)=0\}\rangle,
\end{equation}
where $k[Y]_{\leq d}$ denotes the set of polynomials in $k[Y]$ with degrees not greater than $d$, and $\langle * \rangle$ denotes the ideal in $\ring$ generated by $*$.
When $d=\infty$, $I_{F,d}$ is equal to $I_{F}$ that is defined in Introduction.
One can readily verify that $I_{F,d}$ is a $\sigma$-ideal and furthermore $I_F$ is a maximal $\sigma$-ideal. The fact that $\ring$ is a noetherian ring implies that for sufficiently large $d$, $I_{F,d}$ is a proto-maximal $\sigma$-ideal. Therefore to achieve a proto-maximal $\sigma$-ideal, one only needs to solve the following two problems: $(a)$ Given an integer $d$, how to compute $I_{F,d}$? $(b)$ When is the integer $d$ large enough such that $I_{F,d}$ is proto-maximal?

\subsection{The computation of $I_{F,d}$}
\label{SUBSEC:computation}
In \cite{kauers-zimmermann},  Kauers and Zimmerman presented an algorithm for computing generators for the ideal of algebraic relations
among solutions of linear difference equations with constant coefficients. Their algorithm relies on the fact that one can explicitly write down solutions of the equations of such type.  Here, our task is different. We only compute the ideal generated by algebraic relations with bounded degree, while we are interested in linear difference equations with coefficients in $k$.

We first show that which fundamental matrix $F$ we take in this section. Let $\calS_{\barQ}$ be the difference ring of germs at infinity of $\barQ$ (see Example 1.3 in the page 4 of \cite{put-singer} for the definiton). Let $\rho$ be a nonnegative integer such that $i$ is not a pole of entries of $A$ and $\det(A(i))\neq 0$ if $i\geq \rho$ and $Z_\rho\in \GL_n(\barQ)$. Define an element of $\GL_n(\calS_{\barQ})$, say $\bfZ=(Z_0,Z_1,\cdots)$, as follows: $Z_i=0$ for $0\leq i\leq \rho-1$ and $Z_{i+1}=A(i)Z_i$ for $i\geq \rho$. Define a map
$$
\psi: \ring \rightarrow \calS_{\barQ}
$$
as follows:
$$
  \mbox{for $f\in k$}, \psi(f)=(0,\cdots,0,f(i),f(i+1),\cdots,)\,\,\mbox{and}\,\,\psi(Y)=\bfZ
$$
where $i$ is a nonnegative integer such that $j$ is not a pole of $f$ if $j\geq i$. Proposition 4.1 in the page 45 of \cite{put-singer} states that $\psi$ induces an embedding of $\ring/I$ into $\calS_{\barQ}$, where $I=\ker(\psi)$ that is a maximal $\sigma$-ideal. Let $F$ be the image of $Y$ in $\ring/I$. From this construction, we have that $I_{F,d}=I_{\bfZ,d}$.

The results in Appendix A imply that one can compute an integer $\ell$ such that $I_{F,d}$ has a set of generators consisting of polynomials in $\barQ[x][Y]$ whose degrees in $x$ are not greater than $\ell$. Let $N={d+n^2 \choose d}-1$ and $\bfm_0,\cdots, \bfm_N$ be all elements in $\bZ^{n^2}_{\geq 0}$ with $|\bfm_i|\leq d$. Write $P=\sum c_{j(\ell+1)+i} x^iY^{\bfm_j}$ for $P\in I_{F,d}$, where $Y^{\bfm_i}=\prod y_{j,l}^{m_{i,j,l}}$ with $\bfm_i=(m_{i,j,l})$. We can then reduce the original problem to the following problem: find a basis of the vector space
$$
  U=\left\{\left.(c_0,c_1,\cdots,c_{(N+1)(\ell+1)-1})\in \barQ^{(N+1)(\ell+1)}\,\,\right|\,\, \sum_{i=0}^{\ell}\sum_{j=0}^N c_{j(\ell+1)+i} x^i F^{\bfm_j}=0\right\}.
$$
We are going to solve the latter problem. Observe that $\sigma(x^i F^{\bfm_j})$ is a $k$-linear combination of the monomials $F^{\bfm_0},\cdots, x^i F^{\bfm_j},\cdots, x^\ell F^{\bfm_N}$. Hence there is a nonzero linear difference operator $L$ in $\barQ[x][\partial]$ such that $L(x^i F^{\bfm_j})=0$ for all $i$ with $0\leq i \leq \ell$ and $0\leq j \leq N$. This operator $L$ can be computed using the equation (\ref{EQ:differenceeqns}). Notice that at present, we do not know the ideal $I$ and thus do not know $F$. Fortunately, one can easily compute the sequence solution $\bfZ$, which can be considered as a difference analogue of formal power series solutions of linear differential equations,.

For the convenience, write (\ref{EQ:differenceeqns}) and $L$ in the form of linear recurrence equations
\begin{equation}
\label{EQ:recurrenceeqns1}
   Y_{m+1}=A(m)Y_m, m\geq \rho
\end{equation}
and
\begin{equation}
\label{EQ:recurrenceeqns2}
   L=a_l(m)y_{m+l}+a_{l-1}(m)y_{m+l-1}+\cdots+a_0(m)y_m,\,\,m\geq \nu
\end{equation}
where $\rho$ is a positive integer such that $i$ is not a pole of entries of $A(x)$ and $\det(A(i))\neq 0$ for all $i\geq \rho$, and $
\nu$ is an integer greater than integer roots of $a_l(x)a_0(x)=0$. One easily sees that
\begin{lemma}
\label{LM:recurrence}
Assume that $\{s_\nu, s_{\nu+1},\cdots,\}$ is a solution of (\ref{EQ:recurrenceeqns2}). If
there is a nonnegative integer $j$ such that $s_{\nu+j}=\cdots=s_{\nu+l-1+j}=0$, then $s_i=0$ for all $i\geq \nu$.
\end{lemma}
 Let $\kappa$ be an integer greater than $\rho$ and $\nu$.
Notice that the sequence $\{Z_{\rho},Z_{\rho+1},\cdots,\}$ is a solution of (\ref{EQ:recurrenceeqns1}) and for all $0\leq i\leq \ell$ and $0\leq j \leq N$, the sequence $\{\kappa^i Z_\kappa^{\bfm_j}, (\kappa+1)^i Z_{\kappa+1}^{\bfm_j},\cdots,\}$ is a solution of (\ref{EQ:recurrenceeqns2}). Set
$$P_{\bfc}(x,Y)=\sum_{i=0}^\ell \sum_{j=0}^N c_{j(\ell+1)+i} x^i Y^{\bfm_j}, \mbox{where}\,\,\bfc=(c_0,\cdots,c_{(N+1)(\ell+1)-1})\in \barQ^{(N+1)(\ell+1)}.$$
Then the sequence $\{P_\bfc(\kappa,Z_\kappa),(P_\bfc(\kappa+1,Z_{\kappa+1}),\cdots,\}$ is also a solution of (\ref{EQ:recurrenceeqns2}).
\begin{prop}
\label{PROP:recurrence}
  $\bfc\in U$ if and only if $P_\bfc(i, Z_i)=0$ for all $\kappa\leq i \leq \kappa+l-1$.
\end{prop}
\begin{proof}
 Assume that $\bfc\in U$. Then $P_\bfc(x, F)=0$ and thus $\psi(P_\bfc(x,F))=0$. In other words, there is a positive integer $j$ such that $P_\bfc(i,Z_i)=0$ for all $i\geq j$. Lemma~\ref{LM:recurrence} implies that $P_\bfc(i,Z_i)=0$ for all $\kappa \leq i \leq \kappa+l-1$. Conversely, suppose that $P_\bfc(i,Z_i)=0$ for all $\kappa \leq i \leq \kappa+l-1$. By Lemma~\ref{LM:recurrence} again, $P_\bfc(i,Z_i)=0$ for all $i\geq \kappa$. This implies that $\psi(P_\bfc(x,F))=0$. Equivalently, $P_\bfc(x,F)=0$. Hence $\bfc\in U$.
\end{proof}
The conditions $P_\bfc(i,Z_i)=0$ for all $\kappa\leq i \leq \kappa+l-1$ induce a linear system for $\bfc$. Solving this system, we obtain a basis of $U$.
\begin{algorithm}
\label{ALG:partialrelations}
  Compute a basis of $I_{F,d}$.
  \begin{itemize}
     \item [$(\rmnum{1})$] Using the results in Appendix A, compute an integer $\ell$ such that $I_{F,d}$ has generators consisting of polynomials in $\barQ[x][Y]$ whose degrees in $x$ are not greater than $\ell$.
     \item [$(\rmnum{2})$] Construct a nonzero operator $L$ in $\barQ[x][\partial]$ that annihilates $x^i F^{\bfm_j}$ for all $0\leq i \leq \ell$ and $0\leq j \leq N$, where $\bfm_0,\cdots,\bfm_N$ are all elements in $\bZ_{\geq 0}^{n^2}$ satisfying $|\bfm_i|\leq d$
     \item [$(\rmnum{3})$] Let $\kappa$ be an integer that is greater than both $\rho$ and all integer roots of the leading and trailing coefficients of $L$.
     \item [$(\rmnum{4})$] Compute $Z_\kappa,Z_{\kappa+1},\cdots,Z_{\kappa+l-1}$, where $l=\ord(L)$. Set
         $$ P_\bfc(x, Y)=\sum_{i=0}^\ell \sum_{j=0}^N c_{j(\ell+1)+i}x^i Y^{\bfm_j}, \bfc=(c_0,\cdots,c_{(N+1)(\ell+1)-1}).$$
     Putting $P_\bfc(\kappa, Z_\kappa)=\cdots=P_\bfc(\kappa+l-1, Z_{\kappa+l-1})=0$, we obtain a linear system $\L$ in $c_0,c_1,\cdots,c_{(N+1)(\ell+1)-1}$.
     \item [$(\rmnum{5})$]
        Solve $\L$ and return $\left\{ P_{\bar{\bfc}}(x,Y)\,\,\left|\,\,\bar{\bfc}\in \Zero(\L)\cap \barQ^{(N+1)(\ell+1)} \right.\right\}$.
  \end{itemize}
\end{algorithm}
\begin{example}
\label{EXM:example1}
Consider the Fibonacci numbers $F(n)$. It satisfies that
\[
   \begin{pmatrix} F(n+1)\\ F(n+2)\end{pmatrix}=\begin{pmatrix}0 & 1 \\ 1 & 1\end{pmatrix}\begin{pmatrix} F(n)\\ F(n+1)\end{pmatrix}.
\]
Let $$\bfZ=\left(I_2, \begin{pmatrix}0 & 1 \\ 1 & 1\end{pmatrix}, \begin{pmatrix}0 & 1 \\ 1 & 1\end{pmatrix}^2,\cdots\right).$$
We are going to calculate $I_{\bfZ,2}$. Using the results in Appendix A, one sees that there are generators of $I_{\bfZ,2}$ whose degrees in $x$ are zero. Let $\bfm_0,\cdots,\bfm_{14}$ be all vectors in $\bZ_{\geq 0}^4$ satisfying $|\bfm_i|\leq 2$. Let
$$
  L=\partial^6-4\partial^5+2\partial^4+6\partial^3-4\partial^2-2\partial+1.
$$
Then $L$ annihilates $\bfZ^{\bfm_i}$ for all $0\leq i \leq 14$. Computing the first 6 terms of $\bfZ$, denoted by $Z_i$ for $i=0,\cdots,5$.
Set $\bfc=(c_0,c_1,\cdots,c_{14})$ and let $P_{\bfc}(x,Y)$ be defined as in the step $(d)$. Then
\begin{align*}
  P_{\bfc}(x,Z_0)&=c_0+c_1+c_4+c_5+c_8+c_{14},\\
  P_{\bfc}(x,Z_1)&=c_0+c_2+c_3+c_4+c_9+c_{10}+c_{11}+c_{12}+c_{13}+c_{14},\\
  P_{\bfc}(x,Z_2)&=c_0+c_1+c_2+c_3+2c_4+c_5+c_6+c_7+2c_8+c_9+c_{10}+2c_{11}+c_{12}+2c_{13}+4c_{14},\\
   P_{\bfc}(x,Z_3)&=c_0+c_1+2c_2+2c_3+3c_4+\cdots+4c_{12}+6c_{13}+9c_{14},\\
  P_{\bfc}(x,Z_4)&=c_0+2c_1+3c_2+3c_3+5c_4+\cdots+9c_{12}+15c_{13}+25c_{14},\\
   P_{\bfc}(x,Z_5)&=c_0+3c_1+5c_2+5c_3+8c_4+\cdots+25c_{12}+40c_{13}+64c_{14}.
\end{align*}
Solving the linear system $\{P_\bfc(x,Z_i)|i=0,\cdots,5\}$, one has that
\begin{align*}
  c_0&=0, c_1=-c_4, c_2=-c_4-c_3,c_5=-c_8-c_{14}, \\c_6&=-c_8-2c_{14}-c_7-c_{11}-c_{13}, c_9=-c_{14}-c_{10}-c_{11}-c_{12}-c_{13}.
\end{align*}
From this, one sees that $I_{\bfZ,2}$ is generated by
\begin{align*}
   y_{2,1}-y_{1,2}, \,\,y_{2,2}-y_{1,2}-y_{1,1}.
\end{align*}
\end{example}

\subsection{When is $I_{F,d}$ proto-maximal?}
Let $\tilde{d}$ be as in (\ref{EQ:bounds2}). In this section, we shall show that $I_{F,\tilde{d}}$ is proto-maximal. Before proving this, we first describe some properties of $I_{F,d}$. Note that $I_{F,d}$ is contained in a maximal $\sigma$-ideal $I$. Proposition 1.20 in the page 15 of (\cite{put-singer}) states that $\Zero(I)$ is a trivial $k$-torsor for $\stab(I)$. We show that a similar property holds for $I_{F,d}$, i.e. $\Zero(I_{F,d})$ is a trivial $k$-torsor for $\stab(I_{F,d})$. As $\Zero(I)$ is a trivial $k$-torsor for $\stab(I)$, $\Zero(I)\cap \GL_n(k)\neq \emptyset$. Therefore $\Zero(I_{F,d})\cap\GL_n(k)\neq \emptyset$. For short, we denote by $H_{F,d}$ the stabilizer of $I_{F,d}$.

\begin{lemma}
\label{LM:torsors}
   Let $B$ be an element of  $\Zero(I_{F,d})\cap\GL_n(k)$. Then
   $$
      I_{F,d}=\left\langle \left\{Q(B^{-1}Y)\,\,|\,\, Q\in I_{\barQ}(H_{F,d})\cap \barQ[Y]_{\leq d} \right\}\right\rangle.
   $$
\end{lemma}
\begin{proof}
  Denote by $J_B$ the right-hand side.
  Suppose that $P$ is an element of $k[Y]_{\leq d}$ with $P(F)=0$. Then for each $h\in H_{F,d}$, $P(Yh)\in I_{F,d}$ and therefore $P(Bh)=0$. Write $$P(BY)=\sum_{i=1}^l c_i P_i(Y)$$
   where $P_i(Y)\in \barQ[Y]$ and $c_1,\cdots,c_l$ are linearly independent over $\barQ$. Obviously, for all $i$ with $1\leq i \leq l$, the degree of $P_i(Y)$ is not greater than $d$ and $P_i(h)=0$ for all $h\in H_{F,d}$. In other words, $P_i(Y)\in I_{\barQ}(H_{F,d})\cap \barQ[Y]_{\leq d}$ for all $i=1,\cdots,l$. Hence $P \in J_B$ and then $I_{F,d}\subseteq J_B$.

   Notice that $I_F$ is a maximal $\sigma$-ideal that contains $I_{F,d}$. Let $G=\stab(I_F)$. Observe that the action of $\GLbQ$ on $k[Y]$ preserves the degrees of polynomials. From the definition of $I_{F,d}$, one sees that $G\subseteq H_{F,d}$. Due to Proposition~\ref{PROP:torsors}, $\Zero(I_F)=\bar{B}G(\bar{k})$ for any $\bar{B}\in \Zero(I_F)\cap \GL_n(k)$. This implies that $\Zero(I_F)\subseteq \zero(J_{\bar{B}})$ and thus $J_{\bar{B}}\subseteq I_F$. As $F$ is a zero of $I_F$, it is also a zero of $J_{\bar{B}}$. This together with the fact that $J_{\bar{B}}$ is generated by polynomials in $k[Y]_{\leq d}$ implies that $J_{\bar{B}}\subseteq I_{F,d}$. On the other hand, the previous result shows that $I_{F,d}\subseteq J_{\bar{B}}$. Consequently, $I_{F,d}=J_{\bar{B}}$. It remains to prove that $J_B=J_{\bar{B}}$.

   We first have that $J_{\bar{B}}\subseteq J_B$. Define a $k$-automorphism $\phi$ of $k[Y]$ as follows: $$\phi(P(Y))=P(\bar{B}B^{-1}Y).$$
   Then $\phi(J_{\bar{B}})=J_B$, which implies that $J_B\subseteq \phi(J_B)\subseteq \phi^2(J_B)\subseteq \cdots$. Due to the noetherian property of $\ring$, $J_B=\phi(J_B)$. In the sequel, $J_B=J_{\bar{B}}$.
\end{proof}
The above lemma has the following corollaries.
\begin{cor}
\label{COR:torsors1}
  Let $B$ be an element of $\Zero(I_{F,d})\cap \GL_n(k)$. Then
  $$\Zero(I_{F,d})=BH_{F,d}(\bar{k})$$
  i.e. $\Zero(I_{F,d})$ is a trivial $k$-torsor for $H_{F,d}$.
\end{cor}
\begin{cor}
\label{COR:torsors2}
  Let $I_{\rm irr}$ be an associated prime of $I_{F,d}$. Then
  $\stab(I_{\rm irr})=H_{F,d}^\circ$. Moreover $\Zero(I_{\rm irr})$ is a trivial $k$-torsor for $H_{F,d}^\circ$.
\end{cor}
\begin{proof}
 Let $B$ be an element of $\Zero(I_{\rm irr})\cap \GL_n(k)$.
 By Corollary~\ref{COR:torsors1}, $$\Zero(I_{\rm irr})\cap \GL_n(\bar{k})=BH_i(\bar{k})$$
 where $H_i$ is an irreducible component of $H_{F,d}$. Since $B\in \Zero(I_{\rm irr})$, $H_i=H_{F,d}^\circ$. Thus $\Zero(I_{\rm irr})$ is a trivial $k$-torsor for $H_{F,d}^\circ$. For each $g\in \GL_n(\bar{k})$, one can define an isomorphism $\phi_g$ of $\GL_n(\bar{k})$ given by $\phi_g(Z)=Zg$. As $I_{\rm irr}$ is prime, one can verify that $h\in \stab(I_{\rm irr})$ if and only if $\phi_h(BH_{F,d}^\circ(\bar{k}))=BH_{F,d}^\circ(\bar{k})$. On the other hand, for $h\in \GL_n(\barQ)$, $\phi_h(BH_{F,d}^\circ(\bar{k}))=BH_{F,d}^\circ(\bar{k})$ if and only if $h\in H_{F,d}^\circ$. Therefore $\stab(I_{\rm irr})=H_{F,d}^\circ$.
\end{proof}
\begin{lemma}
\label{LM:property2}
  Suppose that $\bar{F}$ is a fundamental matrix of (\ref{EQ:differenceeqns}). If $\bar{F}$ is a zero of $I_{F,d}$, then $$I_{F,d}=I_{\bar{F},d}.$$
\end{lemma}
\begin{proof}
  From the assumption, one has that $I_{F,d}\subseteq I_{\bar{F},d}$. Observe that $F=\bar{F}h$ for some $h\in \GLbQ$. The definition of $I_{F,d}$ implies that $\phi_h(I_{F,d})=I_{\bar{F},d}$, where the homomorphism $\phi_h$ is given by $\phi_h(Y)=Yh$. The successive applications of $\phi_h$ to $I_{F,d}\subseteq I_{\bar{F},d}$ induce that
  $$
     I_{\bar{F},d}\subseteq \phi_h(I_{\bar{F},d})\subseteq \cdots \subseteq \phi_h^i(I_{\bar{F},d})\subseteq\cdots .
  $$
  The noetherian property of the ring $\ring$ implies that $I_{\bar{F},d}=\phi_h(I_{\bar{F},d})$. Therefore $\phi_h(I_{F,d})=\phi_h(I_{\bar{F},d})$, i.e. $I_{F,d}=I_{\bar{F},d}$.
\end{proof}
\begin{prop}
\label{PROP:protogaloisgroup}
 $\stab(I_{F,\tilde{d}})$ is a proto-group of $\stab(I)$, where $I$ is any maximal $\sigma$-ideal containing $I_{F,\tilde{d}}$. In other words, $I_{F,\tilde{d}}$ is proto-maximal.
\end{prop}
\begin{proof}
   Let $G=\stab(I)$ and $H$ be an algebraic subgroup of $\GLbQ$ that is bounded by $\tilde{d}$ and is a proto-group of $G$. Such $H$ exists by Proposition~\ref{PROP:bounds}.
   Observe that there is a fundamental matrix $\bar{F}$ such that $I=I_{\bar{F}}$. Since $I_{F,\tilde{d}}\subseteq I$, we have that $I_{\bar{F},\tilde{d}}=I_{F,\tilde{d}}$ due to Lemma~\ref{LM:property2}.
   Let $B$ be an element of $\Zero(I)\cap \GL_n(k)$ and let
   $$J=\left\{Q(B^{-1}Y)\,\,|\,\, Q\in I_{\barQ}(H)\cap \barQ[Y]_{\leq \tilde{d}}\right\}.$$
   Since $H$ is bounded by $\tilde{d}$, $\Zero(J)=BH(\bar{k})$. By Proposition~\ref{PROP:torsors}, $\Zero(I)=BG(\bar{k})$. Therefore $J\subseteq I$, because $I$ is radical and $H$ is a proto-group of $G$. Note that $\bar{F}$ is a zero of $I$. One then has that $\bar{F}$ is a zero of $J$.
   This implies that
   $$
     J\subseteq I_{\bar{F},\tilde{d}}=I_{F,\tilde{d}}\subseteq I.
   $$
   The first inclusion holds because $J$ is generated by a set of polynomials in $k[Y]_{\leq \tilde{d}}$.
   Lemma~\ref{LM:torsors} then implies that
   $$
      G\leq H_{F,\tilde{d}} \leq H.
   $$
   Then the proposition follows from Remark~\ref{RM:protogroup}.
\end{proof}
\begin{example}
\label{EXM:example2}Consider
\begin{equation}
\label{EQ:thirdorder}
\sigma\begin{pmatrix}y_1\\y_2 \\y_3 \end{pmatrix}=\begin{pmatrix}0&1&0\\ 0& 0&1 \\ x & 0& 0\end{pmatrix}\begin{pmatrix}y_1\\y_2 \\y_3 \end{pmatrix}.
\end{equation}
Using the method developed in Section 3.1, we can compute a $\sigma$-ideal
\begin{align*}
    I_{\bfZ,2}=\langle &y_{1,1}y_{1,2}, y_{1,1}y_{1,3}, y_{1,1}y_{2,1}, y_{1,1}y_{2,3}, y_{1,1}y_{3,1}, y_{1,1}y_{3,2}, y_{1,2}y_{1,3}, y_{1,2}y_{2,1}, y_{1,2}y_{2,2}, y_{1,2}y_{3,2}, y_{1,2}y_{3,3}, \\ &y_{1,3}y_{2,2}, y_{1,3}y_{2,3}, y_{1,3}y_{3,1}, y_{1,3}y_{3,3}, y_{2,1}y_{2,2}, y_{2,1}y_{2,3}, y_{2,1}y_{3,1}, y_{2,1}y_{3,3}, y_{2,2}y_{2,3}, y_{2,2}y_{3,1}, y_{2,2}y_{3,2},\\& y_{2,3}y_{3,2}, y_{2,3}y_{3,3}, y_{3,1}y_{3,2}, y_{3,1}y_{3,3}, y_{3,2}y_{3,3} \rangle.
\end{align*}
Furthermore, one has that
{\small
\begin{align*}
   \stab(I_{\bfZ,2})=\left\{\left.\begin{pmatrix}\alpha & 0 & 0 \\ 0 & \beta & 0 \\ 0& 0 & \gamma\end{pmatrix}\right| \alpha\beta\gamma\neq 0\right\}\bigcup \left\{\left.\begin{pmatrix}0 & \alpha & 0 \\ 0 & 0& \beta \\ \gamma & 0 & 0\end{pmatrix}\right| \alpha\beta\gamma\neq 0\right\}
   \bigcup \left\{\left.\begin{pmatrix}0 & 0 & \alpha \\ \beta & 0& 0 \\ 0& \gamma & 0\end{pmatrix}\right| \alpha\beta\gamma\neq 0\right\}.
\end{align*}}
Since all elements in $\stab(I_{\bfZ,2})$ is semi-simple, $\stab(I_{\bfZ,2})$ is a proto-Galois group of (\ref{EQ:differenceeqns}) over $k$, i.e. $I_{\bfZ,2}$ is a proto-maximal $\sigma$-ideal.
\end{example}

\section{The computation of a maximal $\sigma^\delta$-ideal}
The results in the previous section enable us to calculate a proto-maximal $\sigma$-ideal. Suppose that we have obtained a proto-maximal $\sigma$-ideal, say $I_{F,\tilde{d}}$. Let $I_{\rm irr}$ be an associated prime of $I_{F,\tilde{d}}$. It can be obtained by the algorithmic solutions of the problem $(P1)$. Since $I_{F,\tilde{d}}$ is a $\sigma$-ideal, $I_{\rm irr}$ is a $\sigma^\delta$-ideal for some positive integer $\delta$. In the following, we will enlarge $I_{\rm irr}$ to obtain a maximal $\sigma^\delta$-ideal. By Corollary~\ref{COR:torsors2}, one sees that for any $B\in \Zero(I_{\rm irr})\cap \GL_n(k)$,
\begin{equation}
\label{EQ:torsor3}
   \Zero(I_{\rm irr})=BH_{F,\tilde{d}}^\circ(\bar{k})\,\,\mbox{and}\,\,\stab(I_{\rm irr})=H_{F,\tilde{d}}^\circ.
\end{equation}
Let $I_\delta$ be a maximal $\sigma^\delta$-ideal that contains $I_{\rm irr}$. Proposition~\ref{PROP:torsors} implies that
$\Zero(I_\delta)$ is a trivial $k$-torsor for $\stab(I_\delta)$, i.e. $\Zero(I_\delta)=BG_\delta(\bar{k})$ where $B\in \Zero(I_\delta)\cap\GL_n(k)$.
Then the equality (\ref{EQ:torsor3}) induces that $\stab(I_\delta)\subseteq H_{F,\tilde{d}}^\circ$. We shall show that $H_{F,\tilde{d}}^\circ$ is a proto-group of $\stab(I_\delta)$.
\begin{lemma}
\label{LM:intersection}
 Let $\tilde{I}$ be a maximal $\sigma^\delta$-ideal and $I=\tilde{I}\cap \sigma(\tilde{I})\cap \cdots\cap \sigma^{\delta-1}(\tilde{I})$. Then
 \begin{itemize}
   \item [$(a)$]
   $I$ is a maximal $\sigma$-ideal.
   \item [$(b)$]
   $[\stab(I):\stab(\tilde{I})]\leq \delta$.
  \end{itemize}
\end{lemma}
\begin{proof}
  $(a)$. Suppose that $\bar{I}$ is a maximal $\sigma$-ideal. Let $J$ be a maximal $\sigma^\delta$-ideal containing $\bar{I}$. Then $\bar{I}\subseteq \cap_{i=0}^{\delta-1} \sigma^i(J)$. On the other hand, $\cap_{i=0}^{\delta-1} \sigma^i(J)$ is a $\sigma$-ideal and thus it is equal to $\bar{I}$.
   For any $g\in \GLbQ$, one can define a $\sigma$-isomorphism $\phi_g$ of $\ring$ given by $\phi_g(Y)=Yg$. From the uniqueness of the Picard Vessiot extensions, one can easily see that there is $g\in \GLbQ$ such that $\phi_g(J)=\tilde{I}$. This implies that $\phi_g(\bar{I})=I$. Hence $I$ is a maximal $\sigma$-ideal.

   $(b)$. Let $G=\stab(I)$ and $\tilde{G}=\stab(\tilde{I})$. Let $B$ be an element of $\Zero(\tilde{I})\cap \GL_n(k)$. Due to Proposition~\ref{PROP:torsors}, one has that
  \begin{equation}
  \label{EQ:eqn1}
      \Zero(I)=BG(\bar{k})\,\,\mbox{and}\,\,\Zero(\tilde{I})=B\tilde{G}(\bar{k}).
  \end{equation}
  Meanwhile, all $\sigma^i(\tilde{I})$ are maximal $\sigma^\delta$-ideals. Hence there are $g_1,\cdots, g_{\delta-1}\in \GLbQ$ such that
  $\phi_{g_i}(\sigma^i(\tilde{I}))=\tilde{I}$. This implies that
  \begin{equation}
   \label{EQ:eqn2}
     \Zero(\sigma^i(\tilde{I}))=B\tilde{G}(\bar{k})g_i,\,\,i=0,1,\cdots, \delta-1.
  \end{equation}
  The equalities (\ref{EQ:eqn1}) and (\ref{EQ:eqn2}) imply that $G=\cup_{i=0}^{\delta-1} \tilde{G}g_i$. In the sequel, $[G:\tilde{G}]\leq \delta.$
\end{proof}
Let $I=I_\delta\cap \sigma(I_\delta)\cap \cdots \cap \sigma^{\delta-1}(I_\delta)$. Then $I_{F,\tilde{d}}\subseteq I$. The above lemma together with Proposition~\ref{PROP:protogaloisgroup} induces that $H_{F,\tilde{d}}$ is a proto-group of $\stab(I)$, i.e.
$$
   \left(H_{F,\tilde{d}}\right)_u \leq \left(\stab(I)\right)^\circ \leq \stab(I) \leq H_{F,\tilde{d}}.
$$
Observe that $\left(H_{F,\tilde{d}}\right)_u=\left(H_{F,\tilde{d}}^\circ\right)_u$. Due to the above lemma again, $\left(\stab(I)\right)^\circ=\left(\stab(I_\delta)\right)^\circ$. Thus
$$
   \left(H_{F,\tilde{d}}^\circ\right)_u \leq \left(\stab(I_\delta)\right)^\circ \leq \stab(I_\delta)\leq H_{F,\tilde{d}}^\circ
$$
i.e. $H_{F,\tilde{d}}^\circ$ is a proto-group of $\stab(I_\delta)$. Proposition~\ref{PROP:kernels} implies that $\stab(I_\delta)$ is the intersection of kernels of some characters of $H_{F,\tilde{d}}^\circ$. This will enable us to construct $I_\delta$. Suppose that $\bar{\chi}_1,\cdots,\bar{\chi}_l$ are characters of $H_{F,\tilde{d}}^\circ$ satisfying
$$
   \ker(\bar{\chi}_1)\cap \cdots\cap \ker(\bar{\chi}_l)=\stab(I_\delta).
$$
Then we have the following lemma.
\begin{lemma}
\label{LM:maximalideal}
   Let $B$ be an element of $\Zero(I_\delta)\cap \GL_n(k)$ and
   $$
       \bS= I_{\rm irr}\cup \{\bar{\chi}_i(B^{-1}Y)-1\,\,|\,\,i=1,\cdots,l\}.
   $$
   Then $\Zero(I_\delta)=\Zero(\bS)$.
\end{lemma}
\begin{proof}
  Let $G_\delta=\stab(I_\delta)$. It suffices to show that $\Zero(\bS)=BG_\delta(\bar{k})$.
  Observe that $B\in \Zero(I_{\rm irr})\cap \GL_n(k)$, which implies that $\Zero(I_{\rm irr})=BH_{F,\tilde{d}}^\circ(\bar{k})$. Suppose that $Z\in BG_\delta(\bar{k})$. As $G_\delta$ is the intersection of kernels of the characters $\bar{\chi}_1,\cdots, \bar{\chi}_l$, one sees that $Z\in \Zero(\bS)$. Conversely, assume that $Z\in \Zero(\bS)$. Then $Z\in \zero(I_{\rm irr})$ and thus $Z=Bh$ for some $h\in H_{F,\tilde{d}}^\circ(\bar{k})$. Meanwhile for each $i=1,\cdots,l$,
  $$\bar{\chi}_i(B^{-1}Z)=\bar{\chi}_i(h)=1.$$
  This implies that $h\in G_\delta(\bar{k})$. Therefore $\Zero(\bS)=BG_\delta(\bar{k})$.
\end{proof}
Proposition~\ref{PROP:torsors} states that invertible elements of $\ring/I_{\rm irr}$ are $k$ multiples of characters of $H_{F,\tilde{d}}^\circ$. Precisely, let $P$ be an invertible element of $\ring/I_{\rm irr}$, then $P=r\chi(B^{-1}Y)$ for some $r\in k$ and some character $\chi$ of $H_{F,\tilde{d}}^\circ$. By the above lemma, to compute $I_\delta$, it suffices to find invertible elements of $\ring/I_{\rm irr}$ that take constant values on $\Zero(I_\delta)$. In the following, we first show that invertible elements of $\ring/I_{\rm irr}$ are actually $\sigma^\delta$-hypergeometric over $k$ and then prove that algebraic relations among $\sigma^\delta$-hypergeometric elements take constant values on $\Zero(I_\delta)$ and enable us to find $I_\delta$. We start with a definition.
\begin{define}
\label{DEF:hypergeometric}
  A nonzero element $P$ of $\ring/I_{\rm irr}$ is said to be $\sigma^\delta$-hypergeometric over $k$ if $P$ is invertible in $\ring/I_{\rm irr}$ and $\sigma^\delta(P)=rP$ for some $r\in k$.
\end{define}
Let $P_1,P_2$ be two $\sigma^\delta$-hypergeometric elements over $k$ of $\ring/I_{\rm irr}$. We say $P_1$ and $P_2$ are {\it similar} if there is $r\in k$ such that $P_1=rP_2$.
 \begin{prop}
 \label{PROP:hypergeometric}
   Let $B$ be an element of $\Zero(I_{\rm irr})\cap \GL_n(k)$ and $\chi$ a character of $H_{F,d}^\circ$ that is represented by an element of $\barQ[Y,1/\det(Y)]$. Then $\chi(B^{-1}Y)$ is a $\sigma^\delta$-hypergeometric element over $k$ of $\ring/I_{\rm irr}$. Furthermore, if $\chi_1$ and $\chi_2$ are two distinct characters, then $\chi_1(B^{-1}Y)$ and $\chi_2(B^{-1}Y)$ are not similar.
 \end{prop}
\begin{proof}
    Obviously, $\chi(B^{-1}Y)$ is invertible in $\ring/I_{\rm irr}$. We first claim that $$\sigma^\delta(B^{-1})A_{\delta} B\in H_{F,\tilde{d}}^\circ(k).$$
    For any $Q\in I_{\barQ}(H_{F,\tilde{d}}^\circ)$, by (\ref{EQ:torsor3}), $Q(B^{-1}Y)\in I_{\rm irr}$.
    As $I_{\rm irr}$ is a $\sigma^\delta$-ideal, one has that $Q(\sigma^{\delta}(B^{-1})A_\delta Y)\in I_{\rm irr}$. Since $B\in \Zero(I_{\rm irr})$, $Q(\sigma^{\delta}(B^{-1})A_\delta B)=0$. This proves the claim.
    Now for any $h\in H_{F,\tilde{d}}^\circ(\bar{k})$,
    $$\chi(\sigma^\delta(B^{-1})A_\delta Bh)-\chi(\sigma^\delta(B^{-1})A_\delta B)\chi(B^{-1}Bh)=0.$$
    This implies that
    $$\chi(\sigma^\delta(B^{-1})A_\delta Y)-\chi(\sigma^\delta(B^{-1})A_\delta B)\chi(B^{-1}Y)\in I_{\rm irr}.$$
    In other words,
    $$\sigma^\delta(\chi(B^{-1}Y))-\chi(\sigma^\delta(B^{-1})A_\delta B)\chi(B^{-1}Y) \in I_{\rm irr},$$
    i.e. $\chi(B^{-1}Y)$ is a $\sigma^\delta$-hypergeometric element over $k$ of $\ring/I_{\rm irr}$. This proves the first assertion.

    Now assume that $\chi_1(B^{-1}Y)-r\chi_2(B^{-1}Y) \in I_{\rm irr}$ for some $r\in k$. Then for any $h\in H_{F,\tilde{d}}^\circ$,
    $$
      \chi_1(h)=\chi_1(B^{-1}Bh)=r\chi_2(B^{-1}Bh)=r\chi_2(h).
    $$
    Particularly, putting $h=I_n$, one then has that $r=1$.
    Thus $\chi_1=\chi_2$, a contradiction.
\end{proof}
Let $\kappa_2$ be as in (\ref{EQ:bounds1}).
Proposition B.17 of \cite{feng} states that $X(H_{F,\tilde{d}}^\circ)$ has generators that are represented by polynomials in $\barQ[Y]_{\leq \kappa_2}$. Denote
$$\H=\left\{P\in k[Y]_{\leq \kappa_2}\left| \begin{array}{c}\mbox{$P$ is $\sigma^\delta$-hypergeometric over $k$ in $\ring/I_{\rm irr}$},\\ P-rQ \notin I_{\rm irr}, \forall \,\,r\in k,\,\,\forall\,\,Q\in \H\setminus\{P\} \end{array}\right.\right\}$$
and
$$\X=\left\{P\in \barQ[Y]_{\leq \kappa_2}\left| \,\,\begin{array}{c} P\in X(H_{F,\tilde{d}}^\circ),\\
            P-Q \notin I_k(H_{F,\tilde{d}}^\circ),\,\,\forall\,\,Q\in \X\setminus\{P\}\end{array}\right.\right\}.$$
Then $\X$ is a set of generators of $X(H_{F,\tilde{d}}^\circ)$.
Furthermore, we have that
\begin{cor}
\label{COR:map}
  There is a bijective map between $\H$ and $\X$.
\end{cor}
\begin{proof}
  Let $B\in \Zero(I_{\rm irr})\cap \GL_n(k)$. We define a map $\tau$ from $\H$ to $X(H_{F,\tilde{d}}^\circ)$ as follows: $\tau(P)=\chi$ where $\chi\in X(H_{F,\tilde{d}}^\circ)$ satisfies $P-r\chi(B^{-1}Y)\in I_{\rm irr}$ for some $r\in k$.
  By Proposition~\ref{PROP:torsors}, for each $P\in \H$, there is a character $\chi$ such that $\tau(P)=\chi$, and such character is unique by Proposition~\ref{PROP:hypergeometric}. Hence $\tau$ is well-defined. From the definition of $\H$, one sees that $\tau$ is injective. We shall prove that $\tau(\H)=\X$. As the map defined in (\ref{EQ:map}) is an isomorphism, one has that
  $$P(BY)/r-\chi(Y)\in I_k\left(H_{F,\tilde{d}}^\circ\right).$$ Hence $\chi(Y)$ can be chosen to be a polynomial in $\barQ[Y]_{\leq \kappa_2}$. That is, $\chi\in \X$. Therefore $\tau(\H)\subseteq \X$. Finally, Proposition~\ref{PROP:hypergeometric} implies that $\tau(\H)=\X$.
\end{proof}
Algorithm~\ref{ALG:hypergeometric} in Appendix B enables us to compute $\H$. Suppose that $\H=\{P_1,\cdots, P_\nu\}$. Let $b_j$ be the certifications of $P_j$, i.e. $\sigma^\delta(P_j)-b_jP_j\in I_{\rm irr}$ for all $1\leq j\leq \nu$. Set
$$
   \calZ=\left\{ (m_1,\cdots, m_\nu)\in \bZ^\nu \,\,\left|\,\,\exists \,\,f\in k^{\times}, \,\,s.t.\,\,\prod_{j=1}^\nu b_j^{m_j}=\frac{\sigma^\delta(f)}{f} \right. \right\}.
$$
$\calZ$ is a finitely generated $\bZ$-module. The solution of the problem $(P4)$ allows us to compute a set of generators of $\calZ$. Assume that $\bfm_1,\cdots, \bfm_\mu$ are generators of $\calZ$ and further suppose that
$$\prod_{j=1}^\nu b_j^{m_{i,j}}=\frac{\sigma^\delta(f_i)}{f_i}$$
 where $f_i\in k$ and $\bfm_i=(m_{i,1},\cdots,m_{i,\mu})$. For each $i=1,\cdots,\mu$, write $\bfm_i=\bfm_i^{+}-\bfm_i^{-}$, where $\bfm_i^{+}, \bfm_i^{-}$ are in $\bZ_{\geq 0}^\nu$ and $\bfm_i^{+}\left(\bfm_i^{-}\right)^T=0$. Denote by $\bP$ the vector $(P_1,\cdots,P_\mu)$ and $\bP^{\bfm}=\prod_{j=1}^{\nu}P_j^{m_j}$ where $\bfm=(m_1,\cdots,m_\nu)$. Let
 $$
    \P=\left\langle I_{\rm irr} \cup \left\{\left.\bP^{\bfm_i^{+}}-f_i \bP^{\bfm_i^{-}}\,\,\right |\,\,i=1,\cdots,\mu\right\}\right\rangle
 $$
It is easy to verify that $\P$ is a $\sigma^\delta$-ideal. Let $I_\delta$ be a maximal $\sigma^\delta$-ideal containing $\P$. Then
\begin{prop}
\label{PROP:maximalideal2}
 $ \Zero(\P)=\Zero(I_\delta)$, {\rm i.e.} $I_\delta=\sqrt{\P}$.
\end{prop}
\begin{proof}
  Let $B$ be an element of $\Zero(I_\delta)\cap \GL_n(k)$ and $G_\delta=\stab(I_\delta)$. Then due to Proposition~\ref{PROP:torsors},
  $$\Zero(I_\delta)=BG_\delta(\bar{k}).$$
  The discussion after Lemma~\ref{LM:intersection} states that $H_{F,\tilde{d}}^\circ$ is a proto-group of $G_\delta$.
  By Proposition~\ref{PROP:kernels}, $G_\delta$ is the intersection of kernels of some characters of $H_{F,\tilde{d}}^\circ$. Let $\Lambda$ be the set of these characters. Observe that $\X$ is a set of generators of $X(H_{F,\tilde{d}}^\circ)$. Suppose that $\bar{\chi}\in \Lambda$. Then
  \begin{equation}
    \label{EQ:formula1}\bar{\chi}=\prod_{i=1}^{\nu} \tau(P_i)^{\alpha_i},
    \end{equation}
  where $\alpha_i\in \bZ$ and $\tau$ is defined as in Corollary~\ref{COR:map}. By Corollary~\ref{COR:map},
for each $i=1,\cdots,\nu$, there is $r_i\in k$ such that
\begin{equation}
 \label{EQ:formula2}
 \tau(P_i)(B^{-1}Y)-r_i P_i \in I_{\rm irr}.
 \end{equation}
 Lemma~\ref{LM:maximalideal} implies that $\bar{\chi}(B^{-1}Y)-1\in I_\delta$. Denote by $\bar{Y}$ the image of $Y$ in $\ring/I_\delta$. Then $\bar{\chi}(B^{-1}\bar{Y})-1=0$. This together with (\ref{EQ:formula1}) and (\ref{EQ:formula2}) induces that
 \begin{equation}
 \label{EQ:formula3}
    \prod_{i=1}^{\nu} r_i^{\alpha_i}P_i^{\alpha_i}(\bar{Y})-1=0.
 \end{equation}
 Applying $\sigma^\delta$ to (\ref{EQ:formula3}), one has that
 \begin{equation}
 \label{EQ:formula4}
    \prod_{i=1}^{\nu} \sigma^\delta\left(r_i^{\alpha_i}\right)b_i^{\alpha_i}P_i^{\alpha_i}(\bar{Y})-1=0.
 \end{equation}
 Combining (\ref{EQ:formula3}) and (\ref{EQ:formula4}), one has that
 $$
     \prod_{i=1}^\nu b_i^{\alpha_i}=\prod_{i=1}^\nu \frac{\sigma^\delta\left(r_i^{-\alpha_i}\right)}{r_i^{-\alpha_i}}.
 $$
 Set $\bfalpha=(\alpha_1,\cdots,\alpha_\nu) \in \bZ^{\nu}$. Then $\bfalpha\in \calZ$. So there are integers $z_1,\cdots, z_\mu$ such that
 $\bfalpha=z_1 \bfm_1+\cdots+z_\mu \bfm_\mu$.

  Let $Z$ be an element of $\Zero(\P)$. Then one has that $\bP^{\bfm_i}(Z)=f_i$ for all $1\leq i \leq \mu$, because $\bP^{\bfm_i^{-}}(Z)\neq 0$.
  By (\ref{EQ:formula1}) and (\ref{EQ:formula2}) again,
 \begin{align*}
   \bar{\chi}(B^{-1}Z)-1&=\prod_{i=1}^\nu \tau(P_i)^{\alpha_i}(B^{-1}Z)-1=\bP^{\bfalpha}(Z)\prod_{i=1}^\nu r_i^{\alpha_i}-1\\ &=\prod_{i=1}^{\mu}\bP^{z_i\bfm_i}(Z)\prod_{i=1}^\nu r_i^{\alpha_i} -1
   =\prod_{i=1}^{\mu}f_i^{z_i}\prod_{i=1}^\nu r_i^{\alpha_i}-1.\\
 \end{align*}
 This implies that the polynomial $\bar{\chi}(B^{-1}X)-1$ takes a constant value on $\Zero(\P)$. Particularly, putting $Z=B$, one has that $\bar{\chi}(B^{-1}B)-1=\prod_{i=1}^{\mu}f_i^{z_i}\prod_{i=1}^\nu r_i^{\alpha_i}-1=0$. In the sequel,  $\bar{\chi}(B^{-1}Z)-1=0$ for all $Z\in \Zero(\P)$. Therefore
 $$
   \Zero(\P)\subseteq \Zero(I_{\rm irr}\cup \{\bar{\chi}(B^{-1}Y)-1\,\,|\,\,\bar{\chi}\in \Lambda\}).
 $$
The former set contains $\Zero(I_\delta)$ and the latter one is equal to $\Zero(I_\delta)$ by Lemma~\ref{LM:maximalideal}. Consequently, $\Zero(\P)=\Zero(I_\delta)$.
\end{proof}
Suppose that $\P$ has been calculated. One can then compute $\sqrt{\P}$ by the methods developed in (\cite{eisenbud-huneke-vasconcelos}, Section 8.7 of \cite{becker-Weispfenning}) and $I=\sqrt{\P}\cap\sigma(\sqrt{\P})\cap\cdots\cap\sigma^{\delta-1}(\sqrt{\P})$ by the algorithm presented in (Section 6.3, page 260 of \cite{becker-Weispfenning}). Then the ideal $I$ is a maximal $\sigma$-ideal by Lemma~\ref{LM:intersection}.
\begin{example}
\label{EXM:example3} (Example~\ref{EXM:example2} continued) We have the following irreducible decomposition:
\begin{align*}
   I_{\bfZ,2}=&\langle y_{1,1}, y_{1,2}, y_{2,2}, y_{2,3}, y_{3,1}, y_{3,3}\rangle\cap \langle y_{1,1}, y_{1,3}, y_{2,1}, y_{2,2}, y_{3,2}, y_{3,3}\rangle \\&
   \cap \langle y_{1,2}, y_{1,3}, y_{2,1}, y_{2,3}, y_{3,1}, y_{3,2}\rangle .
\end{align*}
Set $I_{\rm irr}=\langle y_{1,1}, y_{1,2}, y_{2,2}, y_{2,3}, y_{3,1}, y_{3,3}\rangle$. Then one can easily verify that $I_{\rm irr}$ is a $\sigma^3$-ideal and
$$
  \stab(I_{\rm irr})=\left\{\left.\begin{pmatrix}\alpha & 0 & 0 \\ 0 & \beta & 0 \\ 0& 0 & \gamma\end{pmatrix}\right| \alpha\beta\gamma\neq 0\right\}.
$$
The group of characters of $X(\stab(I_{\rm irr}))$ is generated by $y_{1,1},y_{2,2},y_{3,3}$. Thus we only need to compute $\sigma^3$-hypergeometric elements in $\ring/I_{\rm irr}$ which are represented by linear polynomials in $k[Y]$. By Algorithm~\ref{ALG:hypergeometric}, we have that $y_{1,3}, y_{2,1}, y_{3,2}$ are $\sigma^3$-hypermetric elements of $\ring/I_{\rm irr}$ and further they are not similar in pair. Precisely,
$$
   \sigma^3(y_{1,3})=(x+2)y_{1,3},\,\,\sigma^3(y_{2,1})=x y_{2,1}\,\,\sigma^3(y_{3,2})=(x+1)y_{3,2}.
$$
An easy calculation implies that the only element $(m_1,m_2,m_3)$ in $\bZ^3$ such that $$x^{m_1}(x+1)^{m_2}(x+2)^{m_3}=\sigma^3(f)/f$$ for some $f\in k$ is $(0,0,0)$. This implies that $I_{\rm irr}$ is a maximal $\sigma^3$-ideal.
\end{example}
\section{The algorithm and an example}
We are now ready to present the algorithm for computing the Galois group $\stab(I)$, where $I$ is a maximal $\sigma$-ideal of $\ring$.
\begin{algorithm}
   Input: linear difference equations of the form (\ref{EQ:differenceeqns}).\\
   Output: the Galois group of (\ref{EQ:differenceeqns}) over $k$.
\begin{itemize}
   \item [$(\rmnum{1})$] Compute a proto-maximal ideal $I_{F,\tilde{d}}$ by Algorithm~\ref{ALG:partialrelations}.
   \item [$(\rmnum{2})$] Using algorithms for the problem $(P2)$, compute an associated prime of $I_{F,\tilde{d}}$, denoted by $I_{\rm irr}$. Compute a positive integer $\delta$ such that $I_{\rm irr}$ is a $\sigma^{\delta}$-ideal.
   \item [$(\rmnum{3})$] By Algorithm~\ref{ALG:hypergeometric}, compute $\sigma^\delta$-hypergeometric elements in $\ring/I_{\rm irr}$ that are represented by polynomials in $k[Y]_{\leq \kappa_2}$, and are not similar in pair. Denote them by $P_1,\cdots, P_\nu$.
   \item [$(\rmnum{4})$] Let $b_i$ be the certificates of $P_i$, i.e. $\sigma^{\delta}(P_i)-b_i P_i \in I_{\rm irr}$ where $b_i\in k$ and $i=1,\cdots, \nu$. Using the method for the problem $(P3)$, compute a set of generators of the following $\bZ$-module
       $$
         \calZ=\left\{(z_1,\cdots,z_\nu)\in \bZ^{\nu}\,\,\left|\,\,\exists\,\,f\in k^{\times}, s.t.\,\, \prod_{i=0}^\nu b_i^{z_i}=\frac{\sigma^\delta(f)}{f}\right.\right\}.
       $$
       Denote those generators by $\bfm_1,\cdots, \bfm_\mu$.
   \item [$(\rmnum{5})$] Set $\bP=(P_1,\cdots, P_\nu)$ and find $f_i$, the element in $k$ satisfying $\bP^{\bfm_i}=\sigma^{\delta}(f_i)/f_i$ where $i=1,\cdots,\nu$. Set
   $$
      \P=I_{\rm irr}\cup \left\{\left.\bP^{\bfm_i^{+}}-f_i\bP^{\bfm_i^{-}}\,\,\right|\,\, i=1,\cdots,\mu\right\},
   $$
   where $\bfm_i^{+}, \bfm_i^{-}$ are elements in $\bZ_{\geq 0}^\nu$ satisfying $\bfm_i^{+}-\bfm_i^{-}=\bfm_i$ and $\bfm_i^{+}\left(\bfm_i^{-}\right)^T=0$.
   \item [$(\rmnum{6})$]
      By the algorithms for the problem $(P1)$ and the algorithm presented in (Section 6.3, page 260 of \cite{becker-Weispfenning}), compute $\sqrt{\P}$ and
      $$
         I=\sqrt{\P}\bigcap \sigma\left(\sqrt{\P}\right)\bigcap \cdots \bigcap \sigma^{\delta-1}\left(\sqrt{\P}\right).
      $$
   \item [$(\rmnum{7})$]
      Return $\stab(I)$.
\end{itemize}
\end{algorithm}
The correctness of the algorithm comes from the results presented in the previous sections.
\begin{remark}
    One may suspect that the complexity of the algorithm would be very high, since the integers $\tilde{d}$ and $\kappa_2$ given in (\ref{EQ:bounds1}) and (\ref{EQ:bounds2}) are quite large. These integers guarantee the terminate of the algorithm. However, as shown in Examples~\ref{EXM:example2} and \ref{EXM:example3}, these integers may be much larger than those required in practice.
\end{remark}
In the following, we give an example to illustrate the algorithm.
\begin{example}
\label{EXM:example4}
Consider the following linear difference equations
\begin{equation}
\label{EQ:example4}
\sigma\begin{pmatrix}y_1\\y_2\\y_3\end{pmatrix}=\begin{pmatrix}0 & 1 & 0 \\ x & 0 & 0 \\0 & 0& \frac{1}{x}\end{pmatrix}\begin{pmatrix}y_1\\y_2\\y_3\end{pmatrix}.
\end{equation}
\begin{itemize}
\item [$(\rmnum{1})$]
Using the method developed in Section 3.1, we compute an ideal $\tilde{I}$ generated by polynomials in $I_{F,2}\cap \barQ[Y]$:
\begin{align*}
   \tilde{I}=\langle y_{3,2}, y_{3,1}, y_{2,3}, y_{2,1}y_{2,2}, y_{1,3}, y_{1,2}y_{2,2}, y_{1,1}y_{2,1}, y_{1,1}y_{1,2}\rangle.
\end{align*}
$\tilde{I}$ is a $\sigma$-ideal and
$$
   \stab(\tilde{I})=\left\{\left.\begin{pmatrix}\alpha & 0 & 0 \\ 0 & \beta & 0 \\ 0& 0 & \gamma\end{pmatrix}\right| \alpha\beta\gamma\neq 0\right\}\bigcup \left\{\left.\begin{pmatrix}0 & \alpha & 0 \\ \beta & 0& 0 \\ 0 & 0 & \gamma \end{pmatrix}\right| \alpha\beta\gamma\neq 0\right\}.
$$
As all elements in $\stab(\tilde{I})$ are semi-simple, $\stab(\tilde{I})$ is a proto-maximal $\sigma$-ideal and thus $\tilde{I}$ is a proto-maximal $\sigma$-ideal.
\item [$(\rmnum{2})$]
$\tilde{I}$ is radical and one can compute its irreducible decomposition as follows:
$$
   \tilde{I}=\langle y_{1,1}, y_{1,3}, y_{2,2}, y_{2,3}, y_{3,1}, y_{3,2}\rangle \cap \langle y_{1,2}, y_{1,3}, y_{2,1}, y_{2,3}, y_{3,1}, y_{3,2}\rangle.
$$
Set $I_{\rm irr}=\langle y_{1,1}, y_{1,3}, y_{2,2}, y_{2,3}, y_{3,1}, y_{3,2}\rangle$. Then $I_{\rm irr}$ is a $\sigma^2$-ideal and
$$\stab(I_{\rm irr})=\{\diag(\alpha,\beta,\gamma)|\alpha\beta\gamma\neq 0\}.$$
\item [$(\rmnum{3})$]
Observe that the group of characters of $\stab(I_{\rm irr})$ is generated by linear polynomials. Using Algorithm~\ref{ALG:hypergeometric}, we can find that $\sigma^2$-hypergeometric elements of $\ring/_{\rm irr}$ that are represented by linear polynomials in $k[Y]$ are $y_{1,2}, y_{2,1}, y_{3,3}$. Precisely,
$$
   \sigma^2(y_{1,2})=x y_{1,2}, \,\,\sigma^2(y_{2,1})=(x+1)y_{2,1},\,\,\sigma^2(y_{3,3})=\frac{1}{x(x+1)}y_{3,3}.
$$
\item [$(\rmnum{4})$]
 Set
 $$\calZ=\left\{(z_1,z_2,z_3)\in \bZ^3\,\,\left|\,\,\exists\,\,f\in k^{\times}, s.t.\,\, x^{z_1}(x+1)^{z_2}\left(\frac{1}{x(x+1)}\right)^{z_3}=\frac{\sigma^2(f)}{f}\right.\right\}.$$
 One sees that $\calZ$ is generated by $(1,1,1)$.
 \item [$(\rmnum{5})$]
   Let $\P=\langle I_{\rm irr}\cup \{y_{1,2}y_{2,1}y_{3,3}-1\}\rangle$. One has that $\P$ is a radical ideal and thus is a maximal $\sigma^2$-ideal.
 \item [$(\rmnum{6})$]
     Compute $I=\P \cap \sigma(\P)$. One has that
     \begin{align*}
         I=\langle &y_{3,2}, y_{3,1}, y_{2,3}, y_{2,2}y_{2,1}, y_{1,3}, y_{2,2}y_{1,2}, y_{1,2}y_{2,1}^2 y_{3,3}-y_{2,1}, y_{1,2}^2 y_{2,1}y_{3,3}-y_{1,2}, \\ & y_{1,2}y_{2,1}y_{3,3}+y_{1,1}y_{2,2}y_{3,3}-1, y_{1,1}y_{2,1}, y_{1,1}y_{1,2}\rangle.
     \end{align*}
 \item [$(\rmnum{7})$] Using the Gr\"{o}bner base computation, we have that
 $$
    \stab(I)=\left\{\left.\begin{pmatrix}\alpha & 0 & 0 \\ 0 & \beta & 0 \\ 0& 0 & \gamma\end{pmatrix}\right| \alpha\beta\gamma=1\right\}\bigcup \left\{\left.\begin{pmatrix}0 & \alpha & 0 \\ \beta & 0& 0 \\ 0 & 0 & \gamma \end{pmatrix}\right| \alpha\beta\gamma=1\right\}.
 $$
\end{itemize}
\end{example}
\begin{appendix}
\section{Coefficient bounds for generators of $I_{F,d}$}
Note that $I_{F,d}$ is generated by
$$
  \bS=\{P(Y)\in k[Y]_{\leq d}\,\,|\,\, P(F)=0\},
$$
which is a $k$-vector space of finite dimension. We are going to find coefficient bounds for $\bS$. Precisely, we shall find an integer $\ell$ such that there is a basis of $\bS$ satisfying that coefficients of elements in this basis are of degree $\leq \ell$. Let $N={d+n^2\choose d}$ and $M_1,\cdots, M_N$ be the monomials in entries of $F$ with degrees not greater than $d$. Observe that for a basis of $\bS$, it suffices to find a basis of the following vector space
$$
  \left\{(a_1,\cdots,a_N)\in k^N\,\,\left|\,\,\sum_{i=1}^N a_iM_i=0\right.\right \}.
$$
Furthermore, one sees that $(M_1,\cdots, M_N)^T$ is a solution of linear difference equations, which can be constructed from (\ref{EQ:differenceeqns}). Hence our original problem can be reduced to the following.
\begin{problem}
Assume that $\bfv=(v_1,\cdots, v_n)^T$ is a nonzero solution of (\ref{EQ:differenceeqns}), where the $v_i$ are in some Picard-Vessiot extension ring of $k$. Set
$$
   W=\{(a_1,a_2,\cdots, a_n)\in k^n \,\,|\,\, a_1v_1+\cdots+a_nv_n=0\}.
$$
Find an integer $\ell$ depending on $n$ and $A$, such that $W$ has a basis consisting of vectors whose entries are of degree not greater than $\ell$.
\end{problem}
Without loss of generality, we may assume that $v_1,\cdots,v_r$ are linearly independent over $k$ and
$$
    v_{r+i}=c_{i,1}v_1+\cdots+c_{i,r}v_r, \,\,i=1,\cdots,n-r.
$$
For all $i$ with $1\leq i \leq n-r$, denote $\bfc_i=(c_{i,1},\cdots,c_{i,2},\cdots,c_{i,n})$ where $c_{i,r+i}=-1$ and $c_{i,r+j}=0$ for $1\leq j \leq n-r$ and $j\neq i$. Then $\{\bfc_1,\cdots,\bfc_{n-r}\}$ is a basis of $W$. Actually, for any $\bfa=(a_1,\cdots, a_n)\in W$, we have that $\bfa=-(a_{r+1}\bfc_1+\cdots+a_n\bfc_{n-r}).$ In the following, we are going to find a bound for $\deg(c_{i,j})$, where $i=1,\cdots,n-r, j=1,\cdots,r$. Let $V$ be the solution space of (\ref{EQ:differenceeqns}) and
$$
   \tilde{V}=\{\bfw\in V\,\,|\,\,\bfc_i \bfw^T=0,\,\,\forall \,\,i=1,\cdots, n-r \}.
$$
Then $\tilde{V}$ is a $\barQ$-vector space of finite dimension. Moreover, we have
\begin{lemma}
\label{LM:dimension}
$\dim(\tilde{V})=r$.
\end{lemma}
\begin{proof}
Clearly, $\bfv\in \tilde{V}$. Suppose that $\{\bfv_1,\cdots,\bfv_\mu\}$ is a basis of the vector space over $\barQ$ spanned by the orbit of $\bfv$ under the action of $\gal(K/k)$, the Galois group of (\ref{EQ:differenceeqns}), where $K$ is the ring of fractions of the Picard Vessoit extension of $k$ for (\ref{EQ:differenceeqns}). Then $\bfv_i\in \tilde{V}$ for all $i$ with $1\leq i \leq \mu$. In the sequel, $\dim(\tilde{V})\geq \mu$. In the following, we shall prove that $\mu\geq r$.
Denote the matrix consisting of the first $\mu$ rows of $(\bfv_1,\cdots,\bfv_\mu)$ by $D$ and the remaining one by $U$. For any $\phi \in \gal(K/k)$, there is $[\phi]\in \GL_{\mu}(\barQ)$ such that $\phi(D)=D[\phi]$ and $\phi(U)=U[\phi]$.
Without loss of generality, we may assume that $\det(D)\neq 0$. As for any $\phi\in\gal(K/k)$, $\phi(\det(D))=\det(D)\det([\phi])$. One sees that $\det(D)$ is invertible in $K$ and therefore $D$ is invertible. Now for any $\phi\in\gal(K/k)$,
$$\phi(UD^{-1})=U[\phi][\phi]^{-1}D^{-1}=UD^{-1}. $$
The Galois theory implies that $C=UD^{-1}\in k^{(n-\mu)\times \mu}.$ Set $\tilde{C}=(C, I_{n-\mu})$. Then
$$\tilde{C}\begin{pmatrix} D\\ U\end{pmatrix}=0.$$
Particularly, $\tilde{C}\bfv=0$. This implies that $\dim(W)=n-r\geq n-\mu$ and then $\mu\geq r$. So $\dim(\tilde{V})\geq r$. On the other hand, one has that $\dim(\tilde{V})+n-r\leq n$ and then $\dim(\tilde{V})\leq r$. Hence $\dim(\tilde{V})=r$.
\end{proof}
Assume that $\{\bfv_1=\bfv, \bfv_2,\cdots, \bfv_r\}$ is a basis of $\tilde{V}$ and $M$ is the $n\times r$ matrix consisting of the vectors $\bfv_1,\cdots,\bfv_r$. For $1\leq i_1<\cdots<i_r\leq n$, denote the determinant of
the sub-matrix consisting of the $i_1$-th,$i_2$-th, $\cdots$, $i_r$-th rows of $M$ by $d_{i_1,i_2,\cdots, i_r}$. Then an easy calculation implies that
$$
   d_{i_1,i_2,\cdots, i_r}=b_{i_1,i_2,\cdots, i_r}d_{1,2,\cdots, r}, \mbox{where $b_{i_1,i_2,\cdots, i_r}\in k$}.
$$
In particular,
$b_{1,2,\cdots,j-1,j+1,\cdots,r,r+i}=(-1)^{r-j}c_{i,j}.$
Let $\bfb=(b_{1,2,\cdots,r},\cdots,b_{n-r+1,n-r+2,\cdots,n})^T$.
On the other hand, one can construct from $A$ an invertible matrix $\tilde{A}_r$ with entries in $k$ such that $\bfb d_{1,2,\cdots,r}$ is a solution of $\sigma(Y)=\tilde{A}_r Y$. Notice that the matrix $\tilde{A}_r$ only depends on $A$ and $r$. Moreover, one can easily verify that $d_{1,2,\cdots,r}$ is hypergeometric over $k$. This implies that $\bfb d_{1,2,\cdots,r}$ is a hypergeometric solution. By means of cyclic vector, the system of the form (\ref{EQ:differenceeqns}) can be reduced into a scale linear difference equation. Then algorithms developed in (\cite{cluzeau-van-hoeij, petkosevek}) allow us to find all hypergeometric solutions of (\ref{EQ:differenceeqns}).
Therefore one can find an integer $\ell/2$ such that hypergeometric solutions of $\sigma(Y)=\tilde{A}_rY$ are of the form $\bfw h$ where $h$ is hypergeometric over $k$ and $\bfw$ is a vector whose entries are elements in $k$ with degree not greater than $\ell/2$. Particularly, $\bfb d_{1,2,\cdots,r}=\bar{\bfw}\bar{h}$ where $\bar{\bfw}=(\bar{w}_1,\cdots, \bar{w}_n)\in k^n$ satisfying $\deg(\bar{w}_i)\leq \ell/2$ and $\bar{h}$ is hypgeometric over $k$. Observe that $b_{1,2,\cdots,r}=1$. Then one has that $\bfb=\bar{\bfw}/\bar{w}_1$. Hence entries of $\bfb$ are of degree $\leq \ell$, i.e. $\deg(c_{i,j})\leq \ell$.

In the case that we do not know the dimension of $\tilde{V}$, we can take $r=1,2,\cdots,n$ and construct the corresponding systems $\sigma(Y)=\tilde{A}_1 Y, \cdots, \sigma(Y)=\tilde{A}_n Y$ respectively. Compute all hypergeometric solutions of these systems and let $\ell/2$ be an integer such that these hypergeometric solutions are of the form $\bfw h$ where $h$ is hypergeometric over $k$ and $\bfw$ is a vector whose entries are rational functions in $x$ with degrees not greater than $\ell/2$. Then we have that
$\deg(c_{i,j})\leq \ell$. This solves Problem A.1.

\section{$\sigma^\delta$-Hypergeometric elements}
We shall describe a method to compute $\sigma^\delta$-hypergeometric elements in $\ring/I_{\rm irr}$. In fact, we are not going to calculate all $\sigma^\delta$-hypergeometric elements in $\ring/I_{\rm irr}$. Instead, we only find those $\sigma^\delta$-hypergeometric elements that are represented by polynomials in $k[Y]$ with degrees not greater than $d$ and furthermore that are not similar in pair.
Assume that $\bfm_1,\cdots, \bfm_\ell$ are polynomials in $k[Y]_{\leq d}$ satisfying that $\{\bar{\bfm}_1,\cdots,\bar{\bfm}_\ell\}$ is a $k$-basis of $k[Y]_{\leq d}/(I_{\rm irr}\cap k[Y]_{\leq d})$, where $\bar{\bfm}_i$ is the image of $\bfm_i$. By the Gr\"{o}bner base computation, one can find these $\bfm_i$. As $\sigma^\delta$ preserves the degrees of elements of $k[Y]$, there is $\tilde{A}\in \GL_\ell(k)$ such that
\begin{equation*}
   \sigma^\delta((\bar{\bfm}_1,\bar{\bfm}_2,\cdots, \bar{\bfm}_\ell))=(\bar{\bfm}_1,\bar{\bfm}_2,\cdots,\bar{\bfm}_\ell)\tilde{A}.
\end{equation*}
The invertible matrix $\tilde{A}$ can be constructed from $A$.
Now suppose that $P=\sum c_i \bfm_i$ is a $\sigma^\delta$-hypergeometric element, where $c_i\in k$, i.e. $\sigma^\delta(P)-rP\in I_{\rm irr}$ for some $r\in k$. Then one can verify that $c_1,\cdots,c_\ell$ and $r$ satisfying
$$
\tilde{A}\sigma^\delta \begin{pmatrix}c_1\\ \vdots \\ c_\ell\end{pmatrix}=r\begin{pmatrix}c_1\\ \vdots \\ c_\ell\end{pmatrix}.
$$
Let $h$ be a $\sigma^\delta$-hypergeometric element satisfying $\sigma^\delta(h)=rh$. Then $(c_1,\cdots,c_\ell)^T h$ is
a $\sigma^\delta$-hypergeometric solutions of the following linear difference equations
\begin{equation}
\label{EQ:deltadifferenceeqns}
\sigma^\delta(Y)=\tilde{A}^{-1}Y.
\end{equation}
Consequently, for those $c_1,\cdots,c_\ell$ and $r$, it suffices to find all $\sigma^\delta$-hypergeometric solutions of the above linear difference equations. The algorithms for computing all $\sigma^\delta$-hypergeometric solutions of (\ref{EQ:deltadifferenceeqns}) can be found at (\cite{cluzeau-van-hoeij, petkosevek}). Particularly, one can find $\sigma^\delta$-hypergeometric solutions $\bfc_1 h_1,\cdots, \bfc_l h_l$ that are not similar in pair where $h_1,\cdots,h_l$ are $\sigma^\delta$-hypergeometric and $\bfc_1,\cdots,\bfc_l$ are vectors with entries in $k$. Here two vectors $\bfh_1, \bfh_2$ are said to be similar if $\bfh_1=r\bfh_2$ for some $r\in k^{\times}$. Furthermore, if $\bfh$ is a $\sigma^\delta$-hypergeometric solution of (\ref{EQ:deltadifferenceeqns}), then there is a unique $j$ with $1\leq j \leq l$ satisfying $\bfh=b\bfc_j h_j$ for some $b\in k$. Write $\bfc_i=(c_{i,1},\cdots,c_{i,\ell})$ and set $P_i=\sum_{j=1}^{\ell} c_{i,j}\bfm_j$, where $i=1,2,\cdots,l$. Then $\sigma^\delta(P_i)-r_iP_i\in I_{\rm irr}$ for some $r_i\in k$. It remains to select those $P_i$ that are invertible in $\ring/I_{\rm irr}$. Note that $P_i$ is invertible in $\ring/I_{\rm irr}$ if and only if $\Zero(P_i)\cap \Zero(I_{\rm irr})=\emptyset$. The latter condition can be detected by the Gr\"{o}bner base computation. Precisely, it suffices to decide if $1$ is in the ideal $\langle I_{\rm irr}, P_i\rangle$. The previous results are summarized in the following algorithm.
\begin{algorithm}
\label{ALG:hypergeometric}
Compute all $\sigma^\delta$-hypergeometric elements in $\ring/I_{\rm irr}$ that are represented by polynomials in $k[Y]_{\leq d}$ and are not similar in pair.
\begin{itemize}
\item [$(a)$] Compute a Gr\"{o}bner basis for $I_{\rm irr}\cap k[Y]$ and then find the monomials $\bfm_1,\cdots, \bfm_\ell$ in $k[Y]_{\leq d}$ such that $\{\bar{\bfm}_1,\cdots,\bar{\bfm}_\ell\}$ is a $k$-basis of $k[Y]_{\leq d}/(I_{\rm irr}\cap k[Y])_{\leq d}$, where $\bar{\bfm}_i$ denotes the image of $\bfm_i$ in $\ring/I_{\rm irr}$.
\item [$(b)$] Construct an invertible matrix $\tilde{A}\in \GL_\ell(k)$ such that
   $$
     \sigma^\delta((\bar{\bfm}_1,\bar{\bfm}_2,\cdots, \bar{\bfm}_\ell))=(\bar{\bfm}_1,\bar{\bfm}_2,\cdots,\bar{\bfm}_\ell)\tilde{A}.
   $$
\item [$(c)$] Compute $\sigma^\delta$-hypergeometric elements
  $$
     \sigma^\delta(Y)=\tilde{A}^{-1}Y
  $$
  that are not similar in pair, say $\bfc_1 h_1,\cdots, \bfc_l h_l$, where $h_1,\cdots, h_l$ are $\sigma^\delta$-hypergeometric and $\bfc_1,\cdots,\bfc_l$ are vectors with entries in $k$.
\item [$(d)$]
   Write $\bfc_i=(c_{i,1},\cdots,c_{i,\ell})$ and set $P_i=\sum_{j=1}^\ell c_{i,j}\bfm_j$, where $i=1,2,\cdots,l$.
\item [$(e)$] Decide whether $k[Y,z]=\langle I_{\rm irr}\cap k[Y], P_i, \det(Y)z-1\rangle$ by the Gr\"{o}bner base computation. Return those $P_i$ satisfying $\langle I_{\rm irr}\cap k[Y], P_i, \det(Y)z-1\rangle=k[Y,z]$.
\end{itemize}
\end{algorithm}
\end{appendix}


\begin{thebibliography}{99}
 \bibitem{becker-Weispfenning}
   Thomas Becker Volker Weispfenning, {\em Gr\"{o}bner Bases}, Graduate Texts in Mathematics, Springer-Verlag, New York, Inc., 1993.
 \bibitem{compoint-singer}
  E. Compoint, Michael F. Singer, Computing Galois groups of completely reducible differential equations, {J. Symbolic Comput.} 28 (1999) 473-494.
  \bibitem{cluzeau-van-hoeij}
    T. Cluzeau, M. van Hoeij, Computing hypergeometric solutions of linear recurrence equations, {\em AAECC}, 17, 83-115, 2006.
   \bibitem{cox-little-oshea}
    D.A. Cox, J. Little, D. O'Shea, {Ideals, Varieties, and Algorithms}, Springer-Verlag, New York, 1996.
  \bibitem{derksen-jeandel-koiran}
    Harm Derksena, Emmanuel Jeandelb and Pascal Koiranb, Quantum automata and algebraic groups, {\em J. Symbolic Comput.}, 39,357-371, 2005.
  \bibitem{eisenbud-huneke-vasconcelos}
D. Eisenbud, C. Huneke, W. Vasconcelos, Direct methods for primary decomposition, {Invent. math.} 110 (1992) 207-235.
  \bibitem{feng}
   Ruyong Feng, Hrushovski¡¯s algorithm for computing the Galois group of a linear differential equation, {\em  Advances in Applied Mathematics}, 65, 1-37, 2015.
  \bibitem{gianni-trager-zacharias}
  P. Gianni, B. Trager and G. Zacharias, Gr\"{o}bner bases and primary decomposition of polynomials ideals, {J. Symbolic Comput.} 6 (1988) 149-167.
   \bibitem{hendriks}
    Peter A. Hendriks, An Algorithm determining the difference Galois group of second order linear Difference equations, {\em J. Symbolic Computation}, 26, 445-461, 1998.
  \bibitem{humphreys}
   James E. Humphreys, {\em Linear Algebraic Groups}, Springer-Verlag New York, 1981.
\bibitem{hrushovski}
  Ehud Hrushovski, Computing the Galois group of a linear differential equation, {\em Banach Center Publications}, 58, 97-138, 2002.
  \bibitem{kauers-zimmermann}
    Manuel Kauers and Burkhard Zimmermann, Computing the algebraicrRelations of C-finite sequences and multisequences, {\em J. Symbolic Comput.}, 43(11):787-803, 2008.
  \bibitem{kovacic}
  J.J. Kovacic,  An algorithm for solving second order linear homogeneous differential equations, {J. Symbolic Comput.} 2,3-43,1986.
  \bibitem{magid}
  A. Magid, Finite generation of class groups of rings of invariants, {\em Proc. Amer. Math. Soc.}, 60, 45-48, 1976.
  \bibitem{maier}
  A. Maier, A difference version of Nori¡¯s theorem, {\em Mathematische Annalen},
   359(3-4),759-784,2014.
   \bibitem{petkosevek}
     Mark Petkosevek, Hypergeometric solutions of linear recurrence equations with polynomial coefficients, {\em J. Symbolic Comput.}, 14, 243-264, 1992.
   \bibitem{rettstadt}
    Daniel Rettstadt, On the computation of the differential Galois group, Ph.D. thesis, RWTH Aachen University, 2014.
  \bibitem{rosenlicht}
    M. Rosenlicht, Toroidal algebraic groups, {\em Proc. Amer. Math. Soc.}, 12, 984-988, 1961.
 \bibitem{seidenberg}
  A. Seidenberg, Constructions in algebra, {\em Trans. Amer. Maht. Soc.}, 197, 273-313, 1974.
  \bibitem{singer-ulmer}
   M.F. Singer, F. Ulmer, Galois groups of second and third order linear differential equations, {\em J. Symbolic Comput.} 16 (1993) 9-36.
  \bibitem{singer}
  M.F. Singer, Algebraic relations among solutions of linear differential equations, {\em Trans. Amer. Math. Soc.}, 295, 753-763, 1986.
  \bibitem{singer2}
   M.F. Singer, Algebraic and Algorithmic Aspects of Difference Equations, \href{http://www4.ncsu.edu/~singer/CIMPA/Lecture_Notes.pdf}{Lecture notes at CIMPA conference in Santa Marta Columbia}, 2012.
\bibitem{put-singer}
   Marius van der Put and Michael F. Singer, {\em Galois Theory of
   Difference Equations}, Lecture Notes in Mathematics 1666, Springer-Verlag, Berlin Heidelberg, 1997.
\end{thebibliography}
\end{document}